\documentclass{article}[11pt]
\usepackage{graphicx} 
\usepackage{amsmath}
\usepackage[a4paper, total={6.7in, 10in}]{geometry}
\usepackage{amssymb}
\usepackage{amsthm}
\usepackage{mathabx}
\usepackage{hyperref}
\usepackage{booktabs}

\newcommand{\independent}{\perp\!\!\!\!\!\!\perp} 
\usepackage[numbers,square]{natbib}
\usepackage{tikz}
\usetikzlibrary{positioning, arrows.meta}
\usepackage{amsmath, amsthm, amssymb, amstext, amsfonts}
\usepackage{tabularx} 
\usepackage{hyperref}
	\hypersetup{
   	colorlinks=true,
   	linkcolor=red,
   	filecolor=magenta,      
	urlcolor=green,
 	citecolor=blue,
	}
\usepackage{tikz}
\usetikzlibrary{patterns, arrows.meta}
\usepackage[toc,page]{appendix}

\def\ba{\begin{align}}
\def\ea{\end{align}}
\def\ban{\begin{align*}}
\def\ean{\end{align*}}

\def\be{\begin{eqnarray}}
\def\ee{\end{eqnarray}}
\def\een{\end{eqnarray*}}

\def\bqq{\begin{equation}}
\def\eqq{\end{equation}}
\def\bqqn{\begin{equation*}}
\def\eqqn{\end{equation*}}

\def\sq{$\Box$}

\def\qed{\ifmmode\sq\else{\unskip\nobreak\hfil
\penalty50\hskip1em\null\nobreak\hfil\sq
\parfillskip=0pt\finalhyphendemerits=0\endgraf}\fi\par\medbreak}

\newsavebox{\junk}
\savebox{\junk}[1.6mm]{\hbox{$|\!|\!|$}}

\def\til={{\widetilde =}}

 \def\eq#1/{(\ref{#1})}

\newtheorem{theorem}{Theorem}[section]
\newtheorem{corollary}[theorem]{Corollary}

\newtheorem{lemma}[theorem]{Lemma}
\newtheorem{definition}[theorem]{Definition}

\newtheorem{remark}{Remark}[section]

\def\eq#1/{(\ref{e:#1})}

\def\bdes{\begin{description}}
\def\edes{\end{description}}

\def\notes#1{}

\begin{document}
\title{Fine asymptotics for approximate  simulation\\ of symmetric channels}

\author{Valentinian Lungu\\
University of Cambridge\\
Cambridge, UK\\
\texttt{vml26@cam.ac.uk}%
\and
Aaron B. Wagner\\
Cornell University\\
Ithaca, NY, USA\\
\texttt{wagner@cornell.edu}%
\and
Ioannis Kontoyiannis\\
University of Cambridge \\
Cambridge, UK\\
\texttt{ik355@cam.ac.uk}
}

\maketitle

\begin{abstract}
We study the approximate channel simulation problem in the normal approximation regime. 
Our starting point is the fact established by Cao et al.~\cite{Cao_channel_simulation}, that the minimal codebook size $M_\epsilon^*(W)$ for a fixed error tolerance $\epsilon>0$ is closely related to the $\epsilon$-smooth channel max-information, $I_{\max}^\epsilon(W)$. Using this, we refine
earlier large-blocklength approximations to $M_\epsilon^*(W^{\otimes n})$ for discrete memoryless channels. Specifically, we provide tight third-order expansions for Gallager-symmetric channels and demonstrate that the third-order term is channel-dependent. As in corresponding channel coding expansions, we identify a substantial rate gap between singular and nonsingular channels. Notably, the direction of this gap is reversed relative to channel coding: The third-order term is $-\frac{1}{2}\log n$ for singular channels but $O(\log\log n)$ for nonsingular channels.  \end{abstract}

\section{Introduction}
We consider the problem of approximate \textit{channel simulation}. A sender (Alice) and a receiver (Bob) wish to approximately simulate a target conditional distribution  $W_{Y\mid X}(\cdot\mid x)$, based on a channel $W$ known to both, but for a given input $x$ available only to Alice. Upon observing $x$, Alice losslessly transmits a codeword $f(x, U)$ to Bob, where $U$ denotes their common randomness. Bob then computes $Y = g(f(x, U), U)$, and the key requirement is that the induced conditional distribution $\tilde{W}_{Y\mid X}(\cdot \mid x)$ be $\epsilon$-close to the target distribution $W_{Y\mid X}(\cdot \mid x)$; precise definitions are given in subsequent sections. The fundamental communication limit of this problem is the smallest achievable length of the description $f(x,U)$.

The first formalisation of channel simulation was introduced in~\cite{Bennett_entanglement}, 
with earlier related work appearing in~\cite{Wyner_wiretap, Steinberg_simulation, STEINER2000239}; see~\cite{CheukTing} for an overview of this and other relevant history. More recently,
the evaluation of the fundamental limit of channel simulation was studied in~\cite{Cao_channel_simulation} in the normal approximation regime, where
the smooth channel max-information, $I_{\max}^{\epsilon}(W)$,
was identified as the central quantity governing the behaviour of $M^*_\epsilon(W)$. 
Specifically, let $M_\epsilon^*(W)$ denote the minimal channel simulation codebook size with maximum allowable error $\epsilon$ for a target channel $W$;
the following one-shot bounds were proved in~\cite[Theorem 2 and Theorem 4]{Cao_channel_simulation}: for any
$\delta\in (0, \epsilon]$,
\begin{equation}
\label{eq:Cao_one_shot}
    I_{\max}^\epsilon(W) \leq \log M^*_\epsilon(W)\leq I_{\max}^{\epsilon-\delta}(W) +\log\log \frac{1}{\delta}.
\end{equation}
A link between $I_{\max}^\epsilon(W)$ and channel coding was also established \cite[Theorem 5]{Cao_channel_simulation}; writing $N^*_\epsilon(W)$ for the optimal number of distinct messages that can be transmitted over $W$ with an average probability of error bounded by $\epsilon$, we have, 
for any $0<\delta<\min\{\epsilon,1-\epsilon\}$:
\begin{align}
\log N_{1-\epsilon-\delta}^*(&W) +\log \delta \leq \log M^*_\epsilon(W)
    \leq \log N_{1-\epsilon+\delta}^*(W)+ \log \frac{2}{\delta} + \log\log \frac{4}{\delta^2}.
    \label{eq:Cao_coding}
\end{align}
Using the corresponding channel-coding expansions~\cite{strassen:64b, PPV:10, altug_wagner_symmetric_channels, moulin2016logvolumeoptimalcodesmemoryless, tan:book, Tomamichel_2013}
for $N^*_\epsilon(W^{\otimes n})$  with respect to the memoryless channel
$W^{\otimes n}$, yields,
\begin{equation}
    \log M^*_\epsilon(W^{\otimes n}) = nC(W)+\sqrt{nV_{1-\epsilon}(W)}\Phi^{-1}(1-\epsilon) + O(\log n),
    \label{eq:cao_final_result}
\end{equation}
where $C(W)$ is the capacity of $W$ and $V_{1-\epsilon}(W)$ is the dispersion of the channel. 

As noted in~\cite{Cao_channel_simulation}, while the one-shot bounds in~\eqref{eq:Cao_one_shot} are fairly sharp, the expansion in~\eqref{eq:cao_final_result} only identifies the first- and second-order terms for large $n$. For the channel coding problem, the third-order term has been investigated extensively~\cite{Wagner_RCU, altug_wagner_asymmetric_channels, altug_wagner_symmetric_channels, PPV:10}. Unlike in the cases of
binary hypothesis testing~\cite{Lungu_hyp_test,K-lungu-arxiv:24} and data compression~\cite{varentropy}, the third-order term in channel coding is channel-dependent.
Specifically, suppose $W$ is a {\it Gallager-symmetric} channel. If $W$ is nonsingular (e.g., the BSC) the third-order term is $+\frac{1}{2}\log n$, while if it is singular (e.g., the BEC) it is $O(1)$. 
A similar dichotomy, featuring a $\frac{1}{2}\log n$ rate gap between singular and nonsingular channels, was recently reported for the exact channel synthesis problem in~\cite{Sriramu_redundancy}.

In this work, we consider both singular and nonsingular Gallager-symmetric channels and provide asymptotic expansions for the quantity $I_{\max}^\epsilon(W^{\otimes n})$, resulting in tight third-order approximations for $M_\epsilon^*(W^{\otimes n})$. We again observe a rate gap between the two channel types: the third-order term equals $-\frac{1}{2}\log n$ for singular channels and it is $O(\log \log n)$ for nonsingular channels. Therefore, our results demonstrate that singular channels achieve superior third-order rates, being consistent with the results provided in \cite{Sriramu_redundancy} for the {\em exact} channel synthesis problem and in contrast to  channel coding where nonsingular channels lead to higher (and hence preferable) rates; see Table \ref{tab:third_order_results}.

\begin{table}[h!]
    \centering
    \renewcommand{\arraystretch}{1.5}
    \begin{tabular}{@{}lll@{}}
        \toprule
        \textbf{Problem} & \textbf{Singular channels} & \textbf{Nonsingular channels} \\
        \midrule
        \textbf{$N_\epsilon^*(W^{\otimes n})$} & $nC(W) + \sqrt{n V_\epsilon(W)}\Phi^{-1}(\epsilon) + \boldsymbol{O(1)}$ & $nC(W) + \sqrt{n V_\epsilon(W)}\Phi^{-1}(\epsilon) + \boldsymbol{\frac{1}{2}\log n}+ O(1)$ \\
   $M_\epsilon^*(W^{\otimes n})$ & $ nC(W) -\sqrt{n V_{1-\epsilon}(W)}\Phi^{-1}(\epsilon) \mathbf{\boldsymbol{-\frac{1}{2}\log n}} + O(\log\log n) $ & $nC(W) -\sqrt{n V_{1-\epsilon}(W)}\Phi^{-1}(\epsilon)  + \mathbf{\boldsymbol{O(\log\log n)}}$  \\
        \bottomrule
    \end{tabular}
    \caption{Third-order optimal rates for channel coding and approximate channel simulation in singular vs. nonsingular Gallager-symmetric channels. }
    \label{tab:third_order_results}
\end{table}

\section{Notations, definitions and problem setup}
\subsection{Notations and definitions}
We consider two finite alphabets, $\mathcal{X}$ and $\mathcal{Y}$, representing the input and output spaces, respectively. We denote by $\mathcal{P}(\mathcal{X})$ and $\mathcal{P}(\mathcal{Y})$ the sets of probability mass functions (p.m.f.s) on $\mathcal{X}$ and $\mathcal{Y}$, respectively. Similarly, $\mathcal{P}(\mathcal{Y}\mid \mathcal{X})$ denotes the set of all families $W=(W_{Y|X}(\cdot\mid\cdot))$ of conditional p.m.f.s of $Y$ given $X$. For brevity, we denote the conditional p.m.f.\ $W_{Y|X}(\cdot\mid x)$ of $Y$ given $X = x$ by $W_x(\cdot)$.

We denote $\mathcal{X}(y) := \{x \in \mathcal{X} : W(y\mid x) > 0\}$ the set of input
    symbols that may produce the output $y$ with nonzero probability. Similarly, we let $\mathcal{Y}(x) : = \{y\in \mathcal{Y}: W(y \mid x) >0\}$
    denote support set of $W_x$.
    
    Without loss of generality, we assume throughout the paper that every output symbol $y \in \mathcal{Y}$ can be observed with strictly positive probability for at least one input $x\in \mathcal{X}$; i.e the transition matrix of $W$ does not contain any all-zero columns.

All random variables are denoted by capital letters (e.g., $X$), and their realizations by lowercase letters (e.g., $x$). Subscripts index samples, while superscripts denote sequence lengths: $x^n = [x_1, x_2, \ldots, x_n]$.
For any $M\in \mathbb{N}$, we denote $[M]:=\{1, 2\ldots, M\}$.

We also use a subscript on expectation and variance operators to specify the underlying probability distribution; for instance, for a random variable $Y$ taking values in $\mathcal{Y}$, a given $x \in \mathcal{X}$, and a function $f\colon \mathcal{Y} \to \mathbb{R}$, we write $\mathbb{E}_{W_x}[f(Y)] = \sum_{y\in \mathcal{Y}} f(y) W_x(y)$.  

The total variation (TV) distance between two p.m.f.s $p$ and $q$ on a finite alphabet $\mathcal{X}$ is
\begin{equation*}
    \|p - q\|_{\mathrm{TV}} = \sum_{x \in \mathcal{X} : p(x) \geq q(x)} (p(x) - q(x)).
\end{equation*}
Given a channel $W \in \mathcal{P}(\mathcal{Y}\mid \mathcal{X})$, the capacity of $W$ is
\begin{equation*}
    C(W) = \sup_{p \in \mathcal{P}(\mathcal{X})} I(p; W),
\end{equation*}
where $I(p; W)$ is the mutual information
\begin{equation*}
    I(p; W)  = \sum_{x \in \mathcal{X}} \sum_{y \in \mathcal{Y}} p(x) W(y \mid x) \log \frac{p(x)W(y \mid x)}{p(x)Q(y)},
\end{equation*}
with $p \in \mathcal{P}(\mathcal{X})$ and the induced output distribution ${Q(y) := \sum_{z \in \mathcal{X}} p(z) W(y \mid z)}$ for $y \in \mathcal{Y}$.
Furthermore, we define the $\epsilon$-channel dispersion as
\begin{equation}
    V_\epsilon(W) := \begin{cases}
        \min_{p \in \Pi} U(p, W), & \text{if } \epsilon < 1/2, \\
        \max_{p \in \Pi} U(p, W), & \text{if } \epsilon \geq 1/2,
    \end{cases}
\end{equation}
where $\Pi := \{p \in \mathcal{P}(\mathcal{X}) : I(p; W) = C(W)\}$ is the set of capacity-achieving input distributions and
\begin{equation*}
    U(p, W) := \sum_{x, y}  p(x) W(y \mid x) \left( \log \frac{W(y \mid x)}{Q(y)} - I(p; W) \right)^2.
\end{equation*}

A central quantity in our analysis is the max-divergence,
\begin{equation}
    D_{\max}(p \,\|\, q) := \inf \left\{ \lambda \in \mathbb{R} : p(x) \leq e^\lambda q(x),\text{ } \forall x \in \mathcal{X} \right\},
\end{equation}
associated with the channel max-information $I_{\max}(W)$, defined as
\begin{equation}
    I_{\max}(W) := \sup_{p \in \mathcal{P}(\mathcal{X})} \inf_{q \in \mathcal{P}(\mathcal{Y})} D_{\max}(p \cdot W \,\|\, p \times q).
    \label{eq:ImaxW}
\end{equation}
Finally, letting $\Delta$ denote the
worst-case TV distance on $\mathcal{P}(\mathcal{Y}\mid \mathcal{X})$,
\begin{equation}
\label{def:delta_distance}
    \Delta(\tilde{W}, W) := \sup_{x \in \mathcal{X}} \|\tilde{W}_x - W_x\|_{\mathrm{TV}},
\end{equation}
the $\epsilon$-smooth channel max-information with respect to $\Delta$ is given by
\begin{equation}
    I_{\max}^\epsilon(W) := \inf_{\substack{\tilde{W} \in \mathcal{P}(\mathcal{Y}\mid \mathcal{X}) : \Delta(\tilde{W}, W) \leq \epsilon}} I_{\max}(\tilde{W}).
\end{equation}

In this work, we consider both singular and nonsingular Gallager-symmetric, or simply {\em symmetric},  channels.

\begin{definition}
A discrete channel is {\em symmetric} if the channel output alphabet can be partitioned into
subsets such that, within each subset, the matrix of transition probabilities satisfies the following: each row (resp.
column) is a permutation of each other row (resp. column).
\end{definition} 

\begin{definition}
A symmetric channel $W \in \mathcal{P}(\mathcal{Y}\mid \mathcal{X})$ is  {\em singular} if $W (y\mid x) =W (y\mid z) = : w(y)$ for all $(x, y, z) \in \mathcal{X}\times \mathcal{Y}\times \mathcal{X} $ such that $W(y\mid x)W(y\mid z) >0$. Otherwise, it is {\em nonsingular}.
\end{definition}

\begin{definition}
For a symmetric channel $W \in \mathcal{P}(\mathcal{Y}\mid \mathcal{X})$, the
{\em information density} is defined as
\begin{equation}
    i(x, y) := \log \frac{W(y\mid x)}{Q(y)}, \quad \text{for all } (x, y) \in \mathcal{X}\times \mathcal{Y},
    \label{def:information_density}
\end{equation}
where $Q$ is the {\em unique} capacity-achieving output distribution, given by $Q(y) = \frac{1}{|\mathcal{X}|}\sum_{z\in \mathcal{X}}W(y\mid z)$.
\end{definition}

\begin{remark}
    For any symmetric {\em singular} channel ${W \in \mathcal{P}(\mathcal{Y}\mid \mathcal{X})}$, the information density depends only on the output $Y$; that is,
    \begin{equation*}
        i(x, y) = \mathbb{E}[i(X,Y)\mid Y=y],
    \end{equation*}
    for all $(x,y) \in \mathcal{X}\times \mathcal{Y}$.
    Moreover, it takes the explicit form
    \begin{equation}
        i(x,y)=\begin{cases}
            \log\frac{|\mathcal{X}|}{|\mathcal{X}(y)|}, & \text{if } W(y\mid x)>0, \\
            -\infty, & \text{otherwise}.
         \end{cases}
        \label{eq:singular_info_density}
    \end{equation}
\end{remark}

\begin{remark} 
    For any symmetric channel $W \in \mathcal{P}(\mathcal{Y}\mid \mathcal{X})$, the distribution of the random variable $i(x, Y)$ under the conditional probability measure $W_x$ is identical for all inputs $x \in \mathcal{X}$.
\end{remark}

\begin{remark}
    For symmetric channels, both the minimum and maximum of $U(p,W)$ over $p\in\Pi$ are equal. In particular, $V_\epsilon(W)$ is independent of $\epsilon\in(0,1)$ --- see Lemma \ref{lemma:capacity_dispersion}. For simplicity, we denote this common value by $V(W)$.

    Moreover, for {\em nonsingular} channels, the dispersion $V(W)$ is strictly positive: $V(W)>0$.
\end{remark}

A simple upper bound for the channel max-information is given by the result of the following lemma.

\begin{lemma}
\label{lemma:ub_max_info}
    Let $W \in \mathcal{P}(\mathcal{Y}\mid \mathcal{X})$. The channel max-information is bounded  above as
    \begin{equation}
        I_{\max}(W) \leq \log \sum_{y \in \mathcal{Y}}\max_{x\in \mathcal{X}} W(y\mid x).
        \label{eq:ub_max_channel_info}
    \end{equation}
\end{lemma}
\begin{proof}
    Define $q_m(y) := \frac{1}{c}\max_{x \in \mathcal{X}}W(y\mid x)$, where $c := \sum_{z\in \mathcal{Y}}\max_{x\in \mathcal{X}}W(z\mid x)$ is the normalisation constant. 
    Fix an arbitrary input distribution $p\in \mathcal{P}(\mathcal{X})$. By evaluating the inner infimum in the definition of the channel max-information in~\eqref{eq:ImaxW}
    at the specific distribution $q_m$, we obtain the upper bound
    \begin{align*}
        \inf_{q \in \mathcal{P}(\mathcal{Y})} D_{\max}(p\cdot W \,\|\, p\times q) &\leq D_{\max}(p\cdot W \,\|\, p\times q_m) \\
        &= \inf \left\{\lambda\in \mathbb{R}: p(x)W(y\mid x) \leq e^\lambda p(x)q_m(y), \forall (x,y) \right\}\\
        &= \inf \left\{\lambda\in \mathbb{R}: p(x)W(y\mid x) \leq \frac{e^\lambda}{c} p(x) \max_{x' \in \mathcal{X}} W(y\mid x'), \forall (x,y) \right\}\\
        &\leq \log c,
    \end{align*}
    where the final inequality holds because the condition is satisfied for all $(x,y)$ whenever $\frac{e^\lambda}{c} \geq 1$. Since this bound of $\log c$ depends only on the channel $W$ and holds simultaneously for all input distributions $p \in \mathcal{P}(\mathcal{X})$, taking the supremum over $p$ completes the proof.
\end{proof}

The following corollary is an immediate consequence of Lemma~\ref{lemma:ub_max_info}.

\begin{corollary}
\label{cor}
The $\epsilon$-smooth channel max-information is bounded  above as 
    \begin{equation}
 I^\epsilon_{\max}(W) \leq \log \inf_{\substack{\tilde{W}: \Delta(\tilde{W}, W) \leq \epsilon}}\sum_{y \in \mathcal{Y}}\max_{x\in \mathcal{X}} \tilde{W}(y\mid x).
        \label{eq:ub_epsilon_max_info}
    \end{equation}
\end{corollary}

This motivates the definition of the following quantity, which plays a central role in the subsequent analysis.
\begin{definition}
\label{def:L_quantity}
    For a given channel $W \in \mathcal{P}(\mathcal{Y}\mid \mathcal{X})$, let
    \begin{equation}
        L^\epsilon(W) := \inf_{\substack{\tilde{W}\in \mathcal{P}(\mathcal{Y}\mid \mathcal{X}) : \Delta(\tilde{W}, W) \leq \epsilon}} \sum_{y \in \mathcal{Y}}\max_{x\in \mathcal{X}} \tilde{W}(y\mid x).
        \label{eq:Ldef}
    \end{equation}
\end{definition}
    

\subsection{Problem setup}
We begin by formally introducing the problem of {\em approximate} channel simulation. The goal of this problem is to construct a channel $\tilde{W} \in \mathcal{P}(\mathcal{Y}\mid \mathcal{X})$ that is $\epsilon$-close in the $\Delta$ distance defined in \eqref{def:delta_distance} to a target channel $W\in\mathcal{P}(\mathcal{Y}\mid \mathcal{X})$: 
$$\Delta(\tilde{W}, W) \leq \epsilon,$$
using a randomness-assisted encoder-decoder scheme. 

\begin{definition} Let $\epsilon \in (0,1)$ and consider a random variable $X$ taking values in a finite input set, $\mathcal{X}$, with p.m.f. $p_X \in \mathcal{P}(\mathcal{X})$. We refer to common randomness as a random variable $U \independent X$ taking values in a set $\mathcal{U}$ with a probability distribution $p_U\in \mathcal{P}(\mathcal{U})$.

An encoder, $f$, is a map $f: \mathcal{X} \times \mathcal{U} \mapsto [M],$
    where $M \in \mathbb{N}$ and $[M]$ representing a $M$-sized codebook. Similarly, a decoder, $g$, is a map $g:[M] \times \mathcal{U} \mapsto \mathcal{Y}.$

    We refer to a tuple $(f, g, U)$ as a size-$M$ $\epsilon$-simulation code for the ``target'' channel $W \in \mathcal{P}(\mathcal{Y}\mid \mathcal{X})$ if ${(X, Y) \sim p_X \tilde{W}_{Y\mid X}}$ where $Y  = g(f(X, U), U)$ and $\tilde{W} \in \mathcal{P}(\mathcal{Y}\mid \mathcal{X})$ such that $\Delta(\tilde{W}, W) \leq \epsilon$.
\end{definition}

\begin{definition}
    The minimum channel simulation codebook size with fixed TV distance tolerance ${\epsilon \in (0,1)}$ is 

        $$M_\epsilon^*(W) := \min\{M\in \mathbb{N}: \exists \text{ a size-$M$ $\epsilon$-simulation code for } W \in \mathcal{P}(\mathcal{Y}\mid \mathcal{X})  \}.$$
\end{definition}

As stated in the Introduction, the problem of approximate channel simulation is connected with the problem of channel coding. In fact, \cite{Cao_channel_simulation} show that the minimum channel simulation codebook size can be bounded from below and above as in \eqref{eq:Cao_coding}. For completeness, we define $N_\epsilon^*(W)$ below.
\begin{definition}
    Let $I$ be a uniform random variable taking values in a $N$-sized message set $[N]$ with $N \in \mathbb{N}$. A ($N$) code for channel $W\in \mathcal{P}(\mathcal{Y}\mid\mathcal{X})$, say $(f, \phi)$, consists of an encoder map $f:[N] \mapsto \mathcal{X}$ and a decoder map $\phi:\mathcal{Y} \to [N]$. 
    
    We let $\{\mathcal{A}_m\}_{m=1}^{[N]}$ denote the decoding regions and $\bar{P}_e(f, \phi, W)$ the average error probability of $(f, \phi)$ for channel $W$:

    $$\bar{P}_e(f, \phi):= \frac{1}{N}\sum_{m \in [N]}\sum_{y\in \mathcal{A}_m^c} W(y\mid f(m)).$$

The minimum average error probability
attainable by any $(N)$ code for $W$ is denoted by $\bar{P}_e(N, W)$. Then, for any $\epsilon \in (0,1)$, the largest number of distinct messages that can be transmitted through $W$ with tolerance $\epsilon$ is

$$N_\epsilon^*(W) := \max\{ N\in \mathbb{N}: \bar{P}_e(N, W) \leq \epsilon\}.$$
\end{definition}

Our goal for the next sections is to find approximations for $M^*_\epsilon(W^{\otimes n})$ using the lower bound from \eqref{eq:Cao_coding} and the upper bound from \eqref{eq:Ldef}.


\section{General upper bound: Tightness for nonsingular channels}

The idea behind the proposed channel for the optimisation problem in \eqref{eq:Ldef} is to cap the probabilities of the most likely outputs given a specific input. We set a threshold proportional to the capacity-achieving output distribution, $Q$. The proportionality constant is chosen such that the total ``shaved'' mass is approximately equal to $\epsilon$. Crucially, this choice also bounds above the information density defined in~\eqref{def:information_density}. Informally, we redistribute $\epsilon$ mass from the most informative sequences to the least informative ones; see Figure~\ref{fig:shave_mass_procedure}. 
 
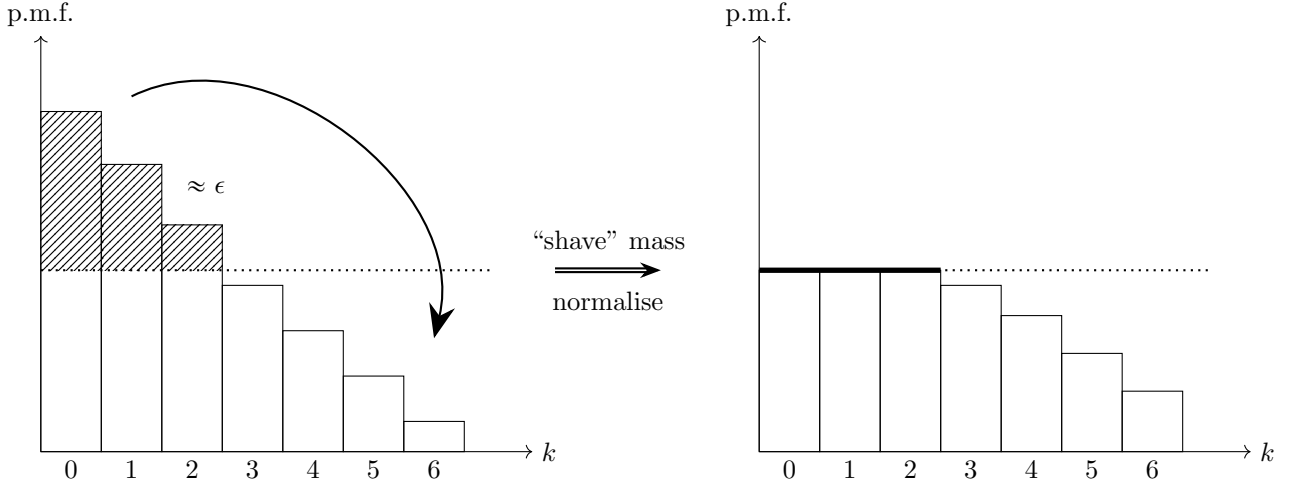
\begin{figure}[!h]
    \centering
    \begin{tikzpicture}[x=1cm, y=1cm]

\draw[->] (0,0) -- (6.5,0) node[right] {$k$};
\draw[->] (0,0) -- (0,5.5) node[above] {p.m.f.};

\foreach \x/\label in {0.4/0, 1.2/1, 2.0/2, 2.8/3, 3.6/4, 4.4/5, 5.2/6} {
    \node[below] at (\x, 0) {\label};
}

\draw (0,0) rectangle (0.8, 4.5);
\draw (0.8,0) rectangle (1.6, 3.8);
\draw (1.6,0) rectangle (2.4, 3.0);
\draw (2.4,0) rectangle (3.2, 2.2);
\draw (3.2,0) rectangle (4.0, 1.6);
\draw (4.0,0) rectangle (4.8, 1.0);
\draw (4.8,0) rectangle (5.6, 0.4);

\path[pattern=north east lines] (0, 2.4) rectangle (0.8, 4.5);
\path[pattern=north east lines] (0.8, 2.4) rectangle (1.6, 3.8);
\path[pattern=north east lines] (1.6, 2.4) rectangle (2.4, 3.0);

\draw[dotted, thick] (0, 2.4) node[left] {} -- (6.0, 2.4);
\node[anchor=west] at (1.8, 3.5) {$\approx \epsilon$};
\draw[-{Stealth[length=4.5mm, width=3.5mm]}, thick, bend left=65] (1.2, 4.7) to node[above, yshift=3mm] {} (5.2, 1.5);

\node at (7.5, 2.8) {``shave'' mass};
\draw[double, -{Stealth}, thick] (6.8, 2.4) -- (8.2, 2.4);
\node at (7.5, 2.0) {normalise};

\begin{scope}[xshift=9.5cm]
\draw[->] (0,0) -- (6.5,0) node[right] {$k$};
\draw[->] (0,0) -- (0,5.5) node[above] {p.m.f.};

\foreach \x/\label in {0.4/0, 1.2/1, 2.0/2, 2.8/3, 3.6/4, 4.4/5, 5.2/6} {
    \node[below] at (\x, 0) {\label};
}

\draw[dotted, thick] (0, 2.4) -- (6.0, 2.4);

\draw (0,0) rectangle (0.8, 2.4);
\draw (0.8,0) rectangle (1.6, 2.4);
\draw (1.6,0) rectangle (2.4, 2.4);

\draw[line width=2pt] (0, 2.4) -- (2.4, 2.4); 

\draw (2.4,0) rectangle (3.2, 2.2);
\draw (3.2,0) rectangle (4.0, 1.8);
\draw (4.0,0) rectangle (4.8, 1.3);
\draw (4.8,0) rectangle (5.6, 0.8);
\end{scope}

\end{tikzpicture}
    \caption{Schematic for the ``mass-shaving'' procedure.}
    \label{fig:shave_mass_procedure}
\end{figure}

\begin{lemma}
\label{lemma:one_shot_ub_general}
    Let $W\in \mathcal{P}(\mathcal{Y}\mid \mathcal{X})$ be a symmetric channel and fix $\epsilon \in (0,1)$. The quantity $L^\epsilon(W)$ defined in Definition~\ref{def:L_quantity} is bounded above as,
    \begin{equation}
        \label{eq:one_shot_ub_general}
        L^\epsilon(W) \leq \inf_{\substack{\gamma \geq -\log r : \\ \mathbb{P}_{W_{x_0}}\left(i(x_0, Y) \geq \gamma\right) \leq \epsilon}} e^\gamma \quad \text{for any arbitrarily chosen } x_0 \in \mathcal{X},
    \end{equation}
    where $r = Q(\mathcal{Y}(x_0))$ and $Q(y) = \frac{1}{|\mathcal{X}|}\sum_{x \in \mathcal{X}}W(y\mid x)$ is the optimal capacity-achieving distribution.
\end{lemma}
\begin{proof}
    We consider the family of conditional p.m.f.s
    \begin{equation}
        \label{def:W_proposal_general}
        \tilde{W}(y\mid x) := \frac{1}{c(M, x)} \min\{W(y\mid x), M Q(y)\}, \quad  (x, y) \in \mathcal{X}\times \mathcal{Y},
    \end{equation}
    where $M > 0$ is a tunable parameter, $Q(y) = \frac{1}{|\mathcal{X}|}\sum_{x\in \mathcal{X}}W(y\mid x)$ is the capacity-achieving output p.m.f.\ for the 
    symmetric channel $W$, and $c(M, x)$
    is the normalising factor. It is bounded from above as follows

    \begin{equation}
        c(M, x) = \sum_{y \in \mathcal{Y}} \min \{ W(y\mid x), MQ(y)\} = \sum_{y \in \mathcal{Y}(x)} \min \{ W(y\mid x), MQ(y)\} \leq M Q\left(\mathcal{Y}(x)\right) =: Mr(x),
    \end{equation}
where $\mathcal{Y}(x)$ is support set of $W_x$. By Lemma~\ref{lemma:c_independence_general}, both  $c(M, x)$ and $r(x)$ are independent of the input symbol $x$; thus, we denote $c(M, x)$  simply as $c(M)$ and $r(x)$ as $r$
    \begin{equation}
    \label{eq:c(M)_general}
        c(M) := \sum_{y\in \mathcal{Y} }\min\{W(y\mid x), M Q(y)\} \leq \min \{Mr, 1\}.
    \end{equation}
    
    Due to the symmetry of the channel, the worst-case TV distance between $\tilde{W}$ and $W$ is independent of the input $x \in \mathcal{X}$. We can bound it as,
    \begin{align}
        \Delta(\tilde{W}, W) &= \max_{x\in \mathcal{X}} \|\tilde{W}_x - W_x \|_{\mathrm{TV}} = \|\tilde{W}_{x_0} - W_{x_0}\|_{\mathrm{TV}}, \quad \text{for any } x_0 \in \mathcal{X} \nonumber\\
        &= \sum_{y\in \mathcal{Y}} \left(W(y\mid x_0) - \tilde{W}(y\mid x_0)\right) \mathbb{I}\left\{ W(y\mid x_0) \geq \tilde{W}(y\mid x_0) \right\} \nonumber\\
        &= \sum_{y\in \mathcal{Y}} \left(W(y\mid x_0) - \tilde{W}(y\mid x_0)\right) \mathbb{I}\left\{ W(y\mid x_0) \geq \frac{1}{c(M)} \min\{W(y\mid x_0), M Q(y)\} \right\} \nonumber\\
        &\leq \sum_{y\in \mathcal{Y}} W(y\mid x_0) \mathbb{I}\left\{ W(y\mid x_0) \geq \frac{M}{c(M)} Q(y) \right\} \label{eq:one_shot_general_costraint_step1}\\
        &= \mathbb{P}_{W_{x_0}}\left( W(Y\mid x_0) \geq \frac{M}{c(M)} Q(Y) \right), \nonumber
    \end{align}
    where~\eqref{eq:one_shot_general_costraint_step1} follows by considering the strictly positive contributions to the TV sum. If $W(y\mid x_0)\leq MQ(y)$, then $\tilde{W}(y\mid x_0)=W(y\mid x_0)/c(M)\geq W(y\mid x_0)$, so the contribution is zero. Otherwise, $\tilde{W}(y\mid x_0)=MQ(y)/c(M)$, and a strictly positive contribution requires $W(y\mid x_0)>MQ(y)/c(M)$. Bounding each such contribution by \(W(y\mid x_0)\) gives~\eqref{eq:one_shot_general_costraint_step1}.
    
    Then, rearranging the terms inside the probability and defining $\gamma(M) := \log \frac{M}{c(M)}$, we obtain
    \begin{equation}
        \label{eq:ub_general_constraint}
        \Delta(\tilde{W}, W) \leq \mathbb{P}_{W_{x_0}}\left( i(x_0, Y) \geq \gamma(M) \right).
    \end{equation}

    Next, evaluating the objective function $L^\epsilon(W)$ for the proposed $\tilde{W}$, gives
    \begin{align}
        \sum_{y\in \mathcal{Y}} \max_{x\in \mathcal{X}} \tilde{W}(y\mid x) &= \frac{1}{c(M)}\sum_{y\in \mathcal{Y}} \max_{x\in \mathcal{X}}\min \{ W(y\mid x), M Q(y)\} \nonumber\\
        &\leq \frac{1}{c(M)}\sum_{y\in \mathcal{Y}} \max_{x\in \mathcal{X}} M Q(y) = \frac{M}{c(M)} = e^{\gamma(M)}, \label{eq:ub_general_obj}
    \end{align}
    where the equality in the second line holds because $\sum_{y\in \mathcal{Y}} Q(y)=1$ and $Q(y)$ is independent of $x$. Therefore, combining the choice of $\tilde{W}$ from~\eqref{def:W_proposal_general} with the bounds in \eqref{eq:ub_general_constraint} and~\eqref{eq:ub_general_obj}, we deduce that for any $M$ satisfying $\mathbb{P}_{W_{x_0}}( i(x_0, Y) \geq \gamma(M) ) \leq \epsilon$, we have $\Delta(\tilde{W}, W) \leq \epsilon$. 
    
    For $S=\mathcal{Y}(x_0)$, let
$$
a=\min_{y\in S}\frac{W(y\mid x_0)}{Q(y)}>0,
\qquad
b=\max_{y\in S}\frac{W(y\mid x_0)}{Q(y)}<\infty.
$$
Since $Q(y)>0$ for any $y \in \mathcal{Y}$ ($W$ does not contain any all-zero columns), $c(M)$ from \eqref{eq:c(M)_general} can be expressed as

$$c(M) = \sum_{y\in \mathcal{Y}}Q(y)\min\left\{M, \frac{W(y\mid x)}{Q(y)}\right\}.$$

Then, $c(M)=Mr$ for $0<M\leq a$, and $c(M)=1$ for $M\geq b$.
Since $c(M)$ is positive and continuous, $M/c(M)$ is continuous, it equals $1/r$ for sufficiently small $M$, and tends to infinity as $M \to \infty$. Together with $M/c(M)\geq 1/r$, this implies that
$$
\{\gamma(M):M>0\}=[-\log r,\infty).$$
Finally, optimising over all valid $\gamma(M) \geq -\log r$ yields
    \begin{equation*}
        L^\epsilon(W) \leq \inf_{\substack{\gamma(M) \geq - \log r : \\ \mathbb{P}_{W_{x_0}}\left(i(x_0, Y) \geq \gamma(M) \right) \leq \epsilon}} e^{\gamma(M)},
    \end{equation*}
    for any arbitrarily chosen $x_0 \in \mathcal{X}$.
\end{proof}

\begin{theorem}
    \label{thm:ub_general}
    Let $\epsilon \in (0,1)$ and $W\in \mathcal{P}(\mathcal{Y}\mid \mathcal{X})$ be a symmetric channel with $V(W) >0$. The $\epsilon$-smooth channel max-information of the memoryless channel $W^{\otimes n}$ is bounded  above as $n \to \infty$
    \begin{equation}
        \label{eq:general_ub}
        I_{\max}^\epsilon(W^{\otimes n}) \leq nC(W) + \sqrt{nV(W)}\Phi^{-1}\left(1-\epsilon +\frac{B(W)}{\sqrt{n}}\right),
    \end{equation}
    where $B(W)$ is a channel-dependent constant determined by the Berry--Esseen bound.
\end{theorem}
\begin{proof}
    
    Since a product of symmetric channels results in a symmetric channel,  the one-shot result of Lemma~\ref{lemma:one_shot_ub_general} applied to $W^{\otimes n}$ gives
    \begin{equation}
        \label{eq:ub_L_n_general}
        L(W^{\otimes n}) \leq \inf_{\substack{\gamma_n \geq -\log r_n : \\ \mathbb{P}_{W_{x^n}}(i(x^n, Y^n)\geq \gamma_n) \leq \epsilon }} e^{\gamma_n}, \quad \text{for any arbitrary } x^n \in \mathcal{X}^n,
    \end{equation}
    where $r_n= Q^n(\mathcal{Y}^n(x^n)) = \prod_{i=1}^n Q(\mathcal{Y}(x_i))$. Following the result of Lemma \ref{lemma:c_independence_general}, $Q(\mathcal{Y}(x_i))$ is independent of $x_i$ implying that $r_n = r^n$, with $r: = Q(\mathcal{Y}(x_0))$ for some arbitrary $x_0\in \mathcal{X}$.
    
    Since the channel is memoryless, the information density tensorizes
    \begin{equation*}
        i(x^n, Y^n) = \log \frac{W^{\otimes n}(Y^n \mid x^n)}{Q^n(Y^n)} = \sum_{j=1}^n \log \frac{W(Y_j \mid x_j)}{Q(Y_j)} = \sum_{j=1}^n i(x_j, Y_j)
    \end{equation*}
    so that its mean is 
    \begin{align}
    \label{eq:mean_info_density_general}
        \mathbb{E}_{W_{x^n}}[i(x^n, Y^n)] = \sum_{j=1}^n \mathbb{E}_{W_{x_j}}[i(x_j, Y_j)]  = n \mathbb{E}_{W_{x_0}}[i(x_0, Y_0)], 
    \end{align}
    where the last step follows from the symmetry of the channel, ensuring that the expectation is independent of the input $x_j$ (evaluated at an arbitrary $x_0\in\mathcal{X}$). Similarly, the variance of the information density is
    \begin{equation}
        \label{eq:variance_info_density_general} 
        \mathrm{Var}_{W_{x^n}}[i(x^n, Y^n)] = \sum_{j=1}^n \mathrm{Var}_{W_{x_j}}[i(x_j, Y_j)] = n \mathrm{Var}_{W_{x_0}}[i(x_0, Y_0)].
    \end{equation}
    Using Lemma~\ref{lemma:capacity_dispersion} for~\eqref{eq:mean_info_density_general} and~\eqref{eq:variance_info_density_general}, we have
    \begin{equation}
        \label{eq:mean_variance_general} 
        \mathbb{E}_{W_{x^n}}[i(x^n, Y^n)] = nC(W) \quad \text{and} \quad \mathrm{Var}_{W_{x^n}}[i(x^n, Y^n)] = nV(W).
    \end{equation}
    Next, we choose the threshold $\gamma_n$ as
    \begin{equation}
        \label{eq:gamma_general}
        \gamma_n = nC(W) + \sqrt{nV(W)}\Phi^{-1}\left(1-\epsilon+ \frac{B(W)}{\sqrt{n}}\right),
    \end{equation}
    where $B(W)$ is the appropriate Berry--Esseen constant for the information density random variable.
    Note that for $n$ large enough, $\gamma_n$ is always well defined; i.e. $\gamma_n > 0$ and $1-\epsilon+\frac{B(W)}{\sqrt{n}} <1$.
    
     Following the results of Lemmas \ref{lemma:capacity_dispersion} and \ref{lemma:capacity_ineq_for_general}, for any $x_0 \in \mathcal{X}$, $C(W) = D(W_{x_0}\|Q) > -\log r $, implying that for any $n$ large enough such that
    $$\sqrt{\frac{V(W)}{n}}\Phi^{-1}\left(1-\epsilon+ \frac{B(W)}{\sqrt{n}}\right)\geq -\log r - C(W),$$
the choice for $\gamma_n$ is admissible for Lemma \ref{lemma:one_shot_ub_general}: $\gamma_n \geq -n\log r = -\log r_n$. Next, we can bound
    \begin{align}
        \mathbb{P}_{W_{x^n}}(i(x^n, Y^n)\geq \gamma_n) &= \mathbb{P}_{W_{x^n}}\left(\frac{\sum_{j=1}^n i(x_j, Y_j) - n C(W)}{\sqrt{nV(W)}} \geq \frac{\gamma_n - nC(W)}{\sqrt{nV(W)}} \right) \nonumber\\
        &= 1 - \mathbb{P}_{W_{x^n}}\left(\frac{\sum_{j=1}^n i(x_j, Y_j) - n C(W)}{\sqrt{nV(W)}} < \frac{\gamma_n - nC(W)}{\sqrt{nV(W)}} \right) \nonumber\\
        &\leq 1 - \Phi\left(\frac{\gamma_n - nC(W)}{\sqrt{nV(W)}} \right) + \frac{B(W)}{\sqrt{n}} \label{eq:berry_general_step2}\\
        &= 1 - \left(1-\epsilon +\frac{B(W)}{\sqrt{n}}\right) +\frac{B(W)}{\sqrt{n}} = \epsilon, \label{eq:berry_general_step3}
    \end{align}
    where~\eqref{eq:berry_general_step2} follows from the Berry--Esseen theorem \cite{petrov-book:95}, and~\eqref{eq:berry_general_step3} is obtained by substituting the value $\gamma_n$ from~\eqref{eq:gamma_general}. Because this choice of $\gamma_n$ satisfies asymptotically the constraint of the optimization problem in~\eqref{eq:ub_L_n_general} , we obtain $L(W^{\otimes n}) \leq e^{\gamma_n}$ for $n \to \infty$. 
    
    Taking the logarithm and applying the bound $I_{\max}^\epsilon(W^{\otimes n}) \leq \log L(W^{\otimes n})$ established in~\eqref{eq:ub_epsilon_max_info}, we obtain the desired result.
\end{proof}

\begin{theorem}
    \label{thm:lb_ub_general}
     Let $\epsilon \in (0,1)$ and $W\in \mathcal{P}(\mathcal{Y}\mid \mathcal{X})$ be a symmetric {\em nonsingular} channel. The minimal simulation codebook size with fixed TV distance tolerance $\epsilon$, $M_\epsilon^*(W^{\otimes n})$, satisfies, as $n\to\infty$,
    \begin{equation}
        \label{eq:lb_ub_general}
        \log M_\epsilon^*(W^{\otimes n}) = nC(W) + \sqrt{n V(W)} \Phi^{-1}(1-\epsilon) + O(\log\log n).
    \end{equation}
\end{theorem}
\begin{proof}
    The result will follow from matching upper and lower bounds for $\log M_\epsilon^*(W^{\otimes n})$. 
    
    We begin with the lower bound.
    Fix $\epsilon\in(0,1)$, and let
$$p=1-\epsilon,\qquad \delta_n=\frac1{\sqrt n},\qquad \alpha_n=p-\delta_n.$$
Choose
$$h=\frac{1}{2}\min\{\epsilon,1-\epsilon\}, \qquad E_\epsilon=[p-h,p+h]\subset(0,1),$$
and write
$$D_\epsilon=\sup_{v\in E_\epsilon}(\Phi^{-1})'(v)<\infty.$$
Then we clearly have $\alpha_n\in E_\epsilon$ eventually, and Lemma~\ref{lem:new} applies.
Define
$$G_\epsilon(n)=nC(W)+\sqrt{nV(W)}\,\Phi^{-1}(1-\epsilon).$$
The mean value theorem gives
$$\sqrt{nV(W)}\,\Phi^{-1}(\alpha_n) \geq \sqrt{nV(W)}\,\Phi^{-1}(p)-\sqrt{V(W)}D_\epsilon,$$
and combining this with the coding-to-simulation inequality and Lemma~\ref{lem:new} we obtain that
\begin{align} 
\log M_\epsilon^*(W^{\otimes n}) 
&\geq \log N_{\alpha_n}^*(W^{\otimes n})+\log\delta_n\nonumber\\
&\geq nC(W)+\sqrt{nV(W)}\,\Phi^{-1}(\alpha_n) +\left(b(W)-\frac12\right)\log n-K(E_\epsilon)
    \nonumber\\ 
&\geq G_\epsilon(n) +\left(b(W)-\frac{1}{2}\right)\log n -K(E_\epsilon)-\sqrt {V(W)} D_\epsilon. 
    \label{eq:toreuse}
\end{align}
Taking $b(W)=\frac{1}{2}$, we obtain,
\begin{equation}
\log M_\epsilon^*(W^{\otimes n})\geq G_\epsilon(n)-K_-,
\label{eq:lb_M_nonsingular}
\end{equation}
with $K_-=K(E_\epsilon)+\sqrt{V(W)}D_\epsilon$. 
        

    For the upper bound, 
    we apply~\cite[Theorem~2]{Cao_channel_simulation},
    \begin{equation*}
        \log M_\epsilon^*(W^{\otimes n}) \leq I_{\max}^{\epsilon-\delta}(W^{\otimes n}) + \log\log \frac{1}{\delta},
    \end{equation*}
    which holds for any $\delta \in (0, \epsilon)$. Set again $\delta = \frac{1}{\sqrt{n}}$ for sufficiently large $n$, and define
$$
\eta_n=\epsilon-\frac1{\sqrt n},\qquad
q_n=1-\epsilon+\frac{B(W)+1}{\sqrt n},
\qquad
\gamma_n^*=nC(W)+\sqrt{nV(W)}\Phi^{-1}(q_n).
$$
Then we have that, eventually, $\eta_n\in(0,1)$, $q_n\in(0,1)$, and
$\frac{\gamma_n^*}{n}\to C(W)>-\log r,$ where the last inequality follows from the results of Lemmas \ref{lemma:capacity_dispersion} and \ref{lemma:capacity_ineq_for_general}.
Hence $\gamma_n^*$ is admissible for Lemma~\ref{lemma:one_shot_ub_general}
 applied to $W^{\otimes n}$. The Berry–Esseen calculation at the end of the proof
 of Theorem~\ref{thm:ub_general}, with error tolerance $\eta_n$, gives
$$
\mathbb P_{W_{x^n}}(i(x^n,Y^n)\geq\gamma_n^*)\leq\eta_n.
$$
Consequently,
$$
I_{\max}^{\eta_n}(W^{\otimes n})\le\gamma_n^*
$$
from which we obtain
    \begin{align}
        \log M_\epsilon^*(W^{\otimes n}) &\leq n C(W) + \sqrt{n V(W)}\Phi^{-1}\left(1-\epsilon+\frac{B(W)+1}{\sqrt{n}} \right) + \log\log \sqrt{n} \nonumber\\
        &\leq n C(W) + \sqrt{n V(W)}\Phi^{-1}\left(1-\epsilon \right) + \log\log n, \label{eq:ub_M_nonsingular}
    \end{align}
    where the last inequality follows from a first-order Taylor expansion of $\Phi^{-1}$, and the fact that $\log\log \sqrt{n} = \log(\frac{1}{2}\log n) = \log\log n - \log 2$, with the constant absorbed into the $O(1)$ term. Combining~\eqref{eq:lb_M_nonsingular} and~\eqref{eq:ub_M_nonsingular} yields the desired result.
\end{proof}

\section{Singular channels}

For the class of singular channels, the upper bound derived in Lemma~\ref{lemma:one_shot_ub_general} is no longer tight. As we establish below, the third-order term in the lower bound for singular channels is exactly $-\frac{1}{2}\log n$, so that the $O(1)$ term obtained via Lemma~\ref{lemma:one_shot_ub_general} is too loose. To close this gap, we introduce a refined one-shot bound for singular channels.

The intuition behind the new proposed simulation channel follows the same principle as in the previous section: We redistribute probability mass from the most informative sequences to the less informative ones. Here, however, instead of truncating the probabilities of the most informative sequences, we set their probabilities exactly to $0$; see Figure~\ref{fig:mass_erasure}.
Consider, for example, the Binary Erasure Channel (BEC)---a canonical singular channel---with erasure probability $\delta < 1/2$. It is computationally easier to simulate output sequences containing more erasures (since they are compatible with more inputs, leading to faster sample acceptance) than sequences with fewer erasures (which are highly informative and tightly restrict the space of valid input sequences). 

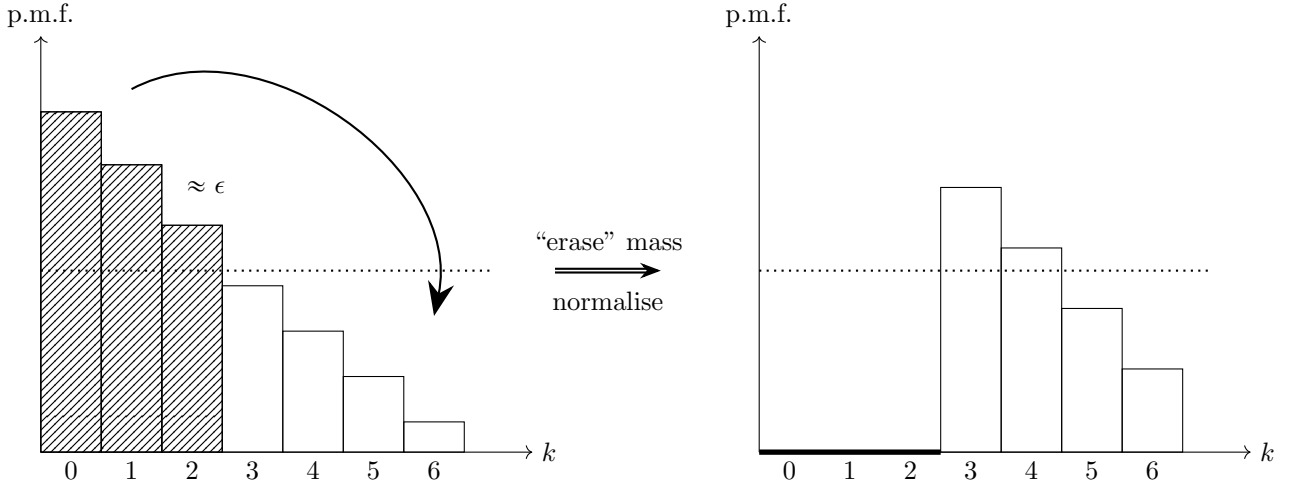
\begin{figure}[h!]
    \centering
    \begin{tikzpicture}[x=1cm, y=1cm]

\draw[->] (0,0) -- (6.5,0) node[right] {$k$};
\draw[->] (0,0) -- (0,5.5) node[above] {p.m.f.};

\foreach \x/\label in {0.4/0, 1.2/1, 2.0/2, 2.8/3, 3.6/4, 4.4/5, 5.2/6} {
    \node[below] at (\x, 0) {\label};
}

\draw (0,0) rectangle (0.8, 4.5);
\draw (0.8,0) rectangle (1.6, 3.8);
\draw (1.6,0) rectangle (2.4, 3.0);
\draw (2.4,0) rectangle (3.2, 2.2);
\draw (3.2,0) rectangle (4.0, 1.6);
\draw (4.0,0) rectangle (4.8, 1.0);
\draw (4.8,0) rectangle (5.6, 0.4);

\path[pattern=north east lines] (0, 0) rectangle (0.8, 4.5);
\path[pattern=north east lines] (0.8, 0) rectangle (1.6, 3.8);
\path[pattern=north east lines] (1.6, 0) rectangle (2.4, 3.0);

\draw (0,0) rectangle (0.8, 4.5);
\draw (0.8,0) rectangle (1.6, 3.8);
\draw (1.6,0) rectangle (2.4, 3.0);

\draw[dotted, thick] (0, 2.4) node[left] {} -- (6.0, 2.4);
\node[anchor=west] at (1.8, 3.5) {$\approx \epsilon$};
\draw[-{Stealth[length=4.5mm, width=3.5mm]}, thick, bend left=65] (1.2, 4.8) to node[above, yshift=3mm] {} (5.2, 1.8);

\node at (7.5, 2.8) {``erase'' mass };
\draw[double, -{Stealth}, thick] (6.8, 2.4) -- (8.2, 2.4);
\node at (7.5, 2.0) {normalise};

\begin{scope}[xshift=9.5cm]
\draw[->] (0,0) -- (6.5,0) node[right] {$k$};
\draw[->] (0,0) -- (0,5.5) node[above] {p.m.f.};

\foreach \x/\label in {0.4/0, 1.2/1, 2.0/2, 2.8/3, 3.6/4, 4.4/5, 5.2/6} {
    \node[below] at (\x, 0) {\label};
}

\draw[dotted, thick] (0, 2.4) -- (6.0, 2.4);

\draw[line width=2pt] (0, 0) -- (2.4, 0); 

\draw (2.4,0) rectangle (3.2, 3.5);
\draw (3.2,0) rectangle (4.0, 2.7);
\draw (4.0,0) rectangle (4.8, 1.9);
\draw (4.8,0) rectangle (5.6, 1.1);
\end{scope}

\end{tikzpicture}
    \caption{Schematic for the ``mass-erasure'' procedure.}
    \label{fig:mass_erasure}
\end{figure}

\begin{lemma}
\label{lemma:one_shot_ub_singular}
    Let $W \in \mathcal{P}(\mathcal{Y}\mid \mathcal{X})$ be a {\em singular} symmetric channel. For any $\epsilon\in (0,1)$ and any arbitrarly chosen $x_0 \in \mathcal{X}$, the quantity $L^\epsilon(W)$ defined in Definition~\ref{def:L_quantity} is bounded  above as 
    \begin{equation}
    \label{eq:ub_one_shot_singular}
        L^\epsilon(W) \leq \inf_{\substack{\gamma \in\mathbb R: \\ \mathbb{P}_{W_{x_0}}\left(i(x_0, Y) \geq \gamma\right) \leq \epsilon}} \frac{1}{1-\epsilon} \mathbb{E}_{W_{x_0}}\left[e^{i(x_0, Y)}\mathbb{I}\left\{i(x_0, Y) < \gamma\right\}\right], \quad \text{with } Y\sim W_{x_0}.
    \end{equation}
\end{lemma}

\begin{proof}
    Let $\gamma \in \mathbb{R}$ be chosen such that

    \begin{equation}
    \label{eq:choice_gamma_singular}
        \mathbb P_{W_{x_0}}(i(x_0,Y)\ge\gamma) \leq \epsilon, \quad \text{with } Y \sim W_{x_0},
    \end{equation}
    for some arbitrary input symbol $x_0\in \mathcal{X}$.
Now consider the family of conditional distributions,
    \begin{equation}
        \label{def:sampled_distr_singular}
        \tilde{W}(y\mid x) := \frac{1}{c(x)} W(y\mid x)\mathbb{I}\{i(x, y) < \gamma\},
    \end{equation}
    where $c(x) := \sum_{y\in \mathcal{Y}}W(y\mid x) \mathbb{I}\{i(x, y) < \gamma\} \leq 1$ is the normalisation constant. Recall that, for a singular channel, the information density is $i(x,y) = \log \frac{|\mathcal{X}|}{|\mathcal{X}(y)|}$ whenever $W(y\mid x) > 0$~\eqref{eq:singular_info_density}. Consequently, $c(x)$ takes the form
    \begin{equation*}
        c(x) = \sum_{y\in \mathcal{Y}}W(y\mid x) \mathbb{I}\left\{\log \frac{|\mathcal{X}|}{|\mathcal{X}(y)|} < \gamma\right\}.
    \end{equation*}
    By~\eqref{eq:c2_constant} proved in Lemma~\ref{lemma:c_independence_general}, this sum is independent of the input $x$, so we simply denote it by $c$. We note that by the choice of $\gamma$ in \eqref{eq:choice_gamma_singular}, $\tilde{W}$ is always well defined because
    \begin{equation}
    \label{eq:lb_c_singular}
    c=\mathbb P_{W_{x_0}}(i(x_0,Y)<\gamma)
  =1-\mathbb P_{W_{x_0}}(i(x_0,Y)\ge\gamma)
  \ge1-\epsilon>0.
  \end{equation}

    Due to the symmetry of the channel, the worst-case TV distance between $\tilde{W}$ and $W$ is independent of the input $x \in \mathcal{X}$, and we can bound it above as follows
    \begin{align}
        \Delta(\tilde{W}, W) &= \max_{x\in \mathcal{X}}\|\tilde{W}_x - W_x \|_{\mathrm{TV}} = \|\tilde{W}_{x_0} - W_{x_0}\|_{\mathrm{TV}} \nonumber\\
        &= \sum_{y\in \mathcal{Y}} \left(W(y\mid x_0) - \tilde{W}(y\mid x_0)\right) \mathbb{I}\left\{W(y\mid x_0) \geq \tilde{W}(y\mid x_0)\right\} \nonumber\\
        &= \sum_{y\in \mathcal{Y}} \left(W(y\mid x_0) - \tilde{W}(y\mid x_0)\right) \mathbb{I}\left\{W(y\mid x_0) \geq \frac{1}{c} W(y\mid x_0)\mathbb{I}\{i(x_0, y) < \gamma\}\right\} \nonumber\\
        &= \sum_{y\in \mathcal{Y}} W(y\mid x_0) \mathbb{I}\left\{i(x_0, y) \geq \gamma\right\} \label{eq:one_shot_singular_costraint_step1}\\
        &= \mathbb{P}_{W_{x_0}}\left(i(x_0, Y) \geq \gamma\right), \label{eq:one_shot_singular_costraint}
    \end{align}
    where~\eqref{eq:one_shot_singular_costraint_step1} holds 
     because outputs satisfying $i(x_0,y)<\gamma$ have $\tilde{W}(y\mid x_0)=W(y\mid x_0)/c\geq W(y\mid x_0)$, and hence contribute zero to the TV sum. Outputs satisfying $i(x_0,y)\geq\gamma$ have $\tilde{W}(y\mid x_0)=0$, and hence contribute $W(y\mid x_0)$. Consequently, the TV distance equals the total probability of the latter region.
    

    Next, we evaluate the objective function for $L^\epsilon(W)$. For a singular channel, if $W(y\mid x) > 0$, it equals a constant $w(y)$ for all valid inputs, and the information density $i(x,y)$ depends only on $y$. Therefore,
    \begin{align}
        \sum_{y\in \mathcal{Y}} \max_{x\in \mathcal{X}} \tilde{W}(y\mid x) &= \sum_{y\in \mathcal{Y}} \max_{x\in \mathcal{X}} \frac{1}{c} W(y\mid x)\mathbb{I}\{i(x, y) < \gamma\} \nonumber\\
        &= \frac{1}{c} \sum_{y\in \mathcal{Y}} w(y)\mathbb{I}\left\{\log \frac{|\mathcal{X}|}{|\mathcal{X}(y)|} < \gamma\right\} \label{eq:obj_singular_step1}\\
        &= \frac{1}{c} \sum_{y\in \mathcal{Y}} Q(y)\frac{|\mathcal{X}|}{|\mathcal{X}(y)|}\mathbb{I}\left\{\log \frac{|\mathcal{X}|}{|\mathcal{X}(y)|} < \gamma\right\}, \label{eq:obj_singular_step1_5}
    \end{align}
    where~\eqref{eq:obj_singular_step1_5} follows from the singular channel property $w(y) = Q(y) \frac{|\mathcal{X}|}{|\mathcal{X}(y)|}$. Expanding the capacity-achieving output distribution $Q(y) = \frac{1}{|\mathcal{X}|}\sum_{x\in \mathcal{X}}W(y\mid x)$, we obtain
    \begin{align}
        \sum_{y\in \mathcal{Y}} \max_{x\in \mathcal{X}} \tilde{W}(y\mid x) &= \frac{1}{c} \sum_{x\in \mathcal{X}} \sum_{y\in \mathcal{Y}} \frac{1}{|\mathcal{X}|}W(y\mid x) \frac{|\mathcal{X}|}{|\mathcal{X}(y)|}\mathbb{I}\left\{\log \frac{|\mathcal{X}|}{|\mathcal{X}(y)|} < \gamma\right\} \label{eq:obj_singular_step2} \\
        &= \frac{1}{c} \sum_{x\in \mathcal{X}} \frac{1}{|\mathcal{X}|} \left(\sum_{y\in \mathcal{Y}} W(y\mid x) e^{i(x, y)}\mathbb{I}\left\{i(x, y) < \gamma\right\} \right) \nonumber\\
        &= \frac{1}{c} \sum_{y\in \mathcal{Y}} W(y\mid x_0) e^{i(x_0, y)}\mathbb{I}\left\{i(x_0, y) < \gamma\right\} \sum_{x\in \mathcal{X}} \frac{1}{|\mathcal{X}|} \label{eq:obj_singular_step3}\\
        &= \frac{1}{c}\mathbb{E}_{W_{x_0}} \left[e^{i(x_0, Y)}\mathbb{I}\left\{ i(x_0, Y) < \gamma\right\} \right], \label{eq:obj_singular_step4}
    \end{align}
    where~\eqref{eq:obj_singular_step3} follows from~\eqref{eq:c3_constant} in  Lemma~\ref{lemma:c_independence_general} in the Appendix, allowing us to replace the inner sum with its value at an arbitrary $x_0 \in \mathcal{X}$. Substituting the lower bound for $c$ from \eqref{eq:lb_c_singular}   into~\eqref{eq:obj_singular_step4} yields
    \begin{equation*}
        \sum_{y\in \mathcal{Y}} \max_{x\in \mathcal{X}} \tilde{W}(y\mid x) \leq \frac{1}{1-\epsilon}\mathbb{E}_{W_{x_0}} \left[e^{i(x_0, Y)}\mathbb{I}\left\{ i(x_0, Y) < \gamma\right\} \right].
    \end{equation*}
    Taking the infimum over all valid thresholds $\gamma$  from \eqref{eq:choice_gamma_singular} yields the desired upper bound on $L^\epsilon(W)$.
\end{proof}

\begin{theorem}
\label{thm:ub_singular}
         Let $\epsilon \in (0,1)$ and $W\in \mathcal{P}(\mathcal{Y}\mid \mathcal{X})$ be a symmetric {\em singular} channel with $V(W)>0$. 
         Let $B(W)$ be the Berry--Esseen constant used in Theorem~\ref{thm:ub_general}, 
        and define
        $$
        T(W):=\mathbb E_{W_{x_0}}
        |i(x_0,Y)-C(W)|^3,
        \qquad
        A(W):=
        \frac1{\sqrt{2\pi V(W)}}
        +\frac{2T(W)}{V(W)^{3/2}},
       $$
       for any $x_0\in{\cal X}$.
         Then, the $\epsilon$-smooth channel max-information of the memoryless channel $W^{\otimes n}$ satisfies, for all $n>(B(W)/\epsilon)^2$,
    \begin{equation}
        \label{eq:singular_ub}
        I_{\max}^\epsilon(W^{\otimes n}) \leq nC(W) + \sqrt{nV(W)}\Phi^{-1}\left(1-\epsilon +\frac{B(W)}{\sqrt{n}}\right) - \frac{1}{2}\log n + \log\Big(\frac{A(W)}{1-\epsilon}\Big). 
    \end{equation}
\end{theorem}

\begin{proof}
    Lemma~\ref{lemma:one_shot_ub_singular} applied to  $W^{\otimes n}$ states that
    \begin{equation}
        \label{eq:ub_L_n_singular1}
        L(W^{\otimes n}) \leq \inf_{\substack{\gamma_n \in \mathbb{R} : \\ \mathbb{P}_{W_{x^n}}(i(x^n, Y^n) \geq \gamma_n)\leq \epsilon}} \frac{1}{1-\epsilon} \mathbb{E}_{W_{x^n}}\left[e^{i(x^n, Y^n)}\mathbb{I}\left\{i(x^n, Y^n) < \gamma_n\right\}\right],
    \end{equation}
    for any $x^n \in \mathcal{X}^n$. 
    As in the part of the proof of Theorem~\ref{thm:ub_general} leading 
    to~\eqref{eq:mean_variance_general}, since the information density $i(x^n, Y^n)$ tensorizes, its mean and variance are given by
    \begin{equation}
    \label{eq:mean_variance_singular}
        \mathbb{E}_{W_{x^n}}[i(x^n, Y^n)] = nC(W) \quad \text{and} \quad \mathrm{Var}_{W_{x^n}}[i(x^n, Y^n)] = nV(W).
    \end{equation}

    We choose the threshold $\gamma_n$ as in~\eqref{eq:gamma_general},
    \begin{equation}
        \label{eq:gamma_singular}
        \gamma_n = nC(W) + \sqrt{nV(W)}\Phi^{-1}\left(1-\epsilon+ \frac{B(W)}{\sqrt{n}}\right),
    \end{equation}
    where $B(W)$ is the Berry--Esseen constant. 
    For $n>(B(W)/\epsilon)^2$,  we have that $1-\epsilon+B(W)/\sqrt n$ belongs in $(0,1)$, so $\gamma_n$ is finite.
        
    With this choice for $\gamma_n$, the tail probability constraint is  satisfied
    \begin{align}
        \mathbb{P}_{W_{x^n}}(i(x^n, Y^n) \geq \gamma_n) &= \mathbb{P}_{W_{x^n}}\left(\frac{\sum_{j=1}^n i(x_j, Y_j) - n C(W)}{\sqrt{nV(W)}} \geq \frac{\gamma_n - nC(W)}{\sqrt{nV(W)}} \right) \nonumber\\
        &= 1 - \mathbb{P}_{W_{x^n}}\left(\frac{\sum_{j=1}^n i(x_j, Y_j) - n C(W)}{\sqrt{nV(W)}} < \frac{\gamma_n - nC(W)}{\sqrt{nV(W)}} \right) \nonumber\\
        &\leq 1 - \Phi\left(\frac{\gamma_n - nC(W)}{\sqrt{nV(W)}} \right) + \frac{B(W)}{\sqrt{n}} \label{eq:berry_singular_step2}\\
        &= 1 - \left(1-\epsilon +\frac{B(W)}{\sqrt{n}}\right) +\frac{B(W)}{\sqrt{n}} = \epsilon, \label{eq:berry_singular_step3}
    \end{align}
    where~\eqref{eq:berry_singular_step2} follows from the Berry--Esseen theorem, and~\eqref{eq:berry_singular_step3} is obtained by substituting our choice of $\gamma_n$ from~\eqref{eq:gamma_singular}. After rearranging and using the tensorization property, the objective in~\eqref{eq:ub_L_n_singular1} can be bounded from above using~\cite[Lemma~2]{altug_wagner_symmetric_channels} as,
    \begin{align}
        \frac{1}{1-\epsilon} \mathbb{E}_{W_{x^n}}\left[e^{i(x^n, Y^n)}\mathbb{I}\left\{i(x^n, Y^n) < \gamma_n\right\}\right] &= \frac{1}{1-\epsilon} \mathbb{E}_{W_{x^n}}\left[e^{\sum_{j=1}^n i(x_j, Y_j)}\mathbb{I}\left\{\sum_{j=1}^n i(x_j, Y_j) < \gamma_n\right\}\right] \nonumber\\
        &\leq \frac{1}{1-\epsilon}\cdot \frac{e^{\gamma_n}}{\sqrt{n}}\left(\frac{1}{\sqrt{2\pi V(W)}} + \frac{2T(W)}{V(W)^{3/2}}\right).
        \label{eq:tocombine}
    \end{align}
    Then, combining~\eqref{eq:ub_L_n_singular1} and~\eqref{eq:tocombine} gives,
$$
L(W^{\otimes n})
\leq
\frac{A(W)}{1-\epsilon}\,
\frac{e^{\gamma_n}}{\sqrt n},
$$
and taking logarithms and applying Corollary~\ref{cor} we get
$$
I_{\max}^{\epsilon}(W^{\otimes n})
\le
\gamma_n-\frac{1}{2}\log n
+\log\left(\frac{A(W)}{1-\epsilon}\right).
$$
Substituting the value of $\gamma_n$  proves the stated bound.

\end{proof}

\begin{theorem}
    Let $\epsilon \in (0,1)$ and $W\in \mathcal{P}(\mathcal{Y}\mid \mathcal{X})$ be a symmetric {\em singular} channel with $V(W) >0$. The minimal simulation codebook size $M_\epsilon^*(W^{\otimes n})$ with fixed TV distance tolerance $\epsilon$ satisfies, as $n\to\infty$,
    \begin{equation}
        \label{eq:lb_ub_singular}
        \log M_\epsilon^*(W^{\otimes n}) = nC(W) + \sqrt{n V(W)} \Phi^{-1}(1-\epsilon) - \frac{1}{2}\log n + O(\log\log n).
    \end{equation}
\end{theorem}

\begin{proof}
    We begin with the lower bound. Arguing exactly as in the proof of Theorem~\ref{thm:lb_ub_general} up to the point leading to~\eqref{eq:toreuse}, and taking $b(W)=0$ this time,  we obtain,
    in the notation of that proof,
    \begin{equation}
  \log M_\epsilon^*(W^{\otimes n})
  \geq G_\epsilon(n)-\frac{1}{2}\log n-K_-.
  \label{eq:lb_M_singular}
    \end{equation}

    For the upper bound, we apply~\cite[Theorem~2]{Cao_channel_simulation},
    \begin{equation*}
        \log M_\epsilon^*(W^{\otimes n}) \leq I_{\max}^{\epsilon-\delta}(W^{\otimes n}) + \log\log \frac{1}{\delta},
    \end{equation*}
    which holds for any $\delta \in (0, \epsilon)$.

With $\eta_n$ and $\gamma_n^*$ as in the proof of Theorem~\ref{thm:lb_ub_general}, we have
\begin{align*}
I_{\max}^{\eta_n}(W^{\otimes n})
&\leq \gamma_n^*-\frac{1}{2}\log n
    +\log\frac{A(W)}{1-\eta_n}\\
&=\gamma_n^*-\frac12\log n
    +\log\frac{A(W)}{1-\epsilon+n^{-1/2}}\\
&\leq\gamma_n^*-\frac12\log n
    +\log\frac{A(W)}{1-\epsilon}.
\end{align*}    
    
    So now setting $\delta = \frac{1}{\sqrt{n}}$ for $n$ sufficiently large yields
    \begin{align}
        \log M_\epsilon^*(W^{\otimes n}) &\leq n C(W) + \sqrt{n V(W)}\Phi^{-1}\left(1-\epsilon+\frac{B(W)+1}{\sqrt{n}} \right) - \frac{1}{2}\log n + \log\log \sqrt{n}  \nonumber\\
        &\leq n C(W) + \sqrt{n V(W)}\Phi^{-1}\left(1-\epsilon \right) - \frac{1}{2}\log n + \log\log n +O(1), \label{eq:ub_M_singular}
    \end{align}
    where the last inequality again follows by a first-order Taylor expansion of $\Phi^{-1}$, and by absorbing the constant from $\log\log \sqrt{n}$ into the $O(1)$ term. Combining~\eqref{eq:lb_M_singular} and~\eqref{eq:ub_M_singular} completes the proof.
\end{proof}

\section{Conclusions}
In this work we provide new upper bounds for the $\epsilon$-smooth channel max-information for Gallager-symmetric channels. Combined with the results of~\cite{Cao_channel_simulation}, these lead to finer expansions for the minimal simulation codebook size with fixed TV distance tolerance $\epsilon$. In particular, in the normal approximation regime, we
precisely identify the third-order term in these expansions. This reveals that, as in channel coding, there is a significant rate gap between singular and nonsingular channels. 

Extending these results to non-symmetric channels is an open and interesting avenue worth exploring in the future. 

\section*{Statement on the use of AI}
Gemini Pro was used for generating Figures \ref{fig:shave_mass_procedure} and \ref{fig:mass_erasure}. Moreover, it was used for arguing the claims of Lemma \ref{lemma:c_independence_general}.  All the proofs appearing in this
paper were written by the authors.
\bibliographystyle{plainnat}
\bibliography{thesis}
\newpage

\appendix
\section{Appendix: Additional lemmas}
\begin{lemma}
\label{lemma:c_independence_general}
    Let $W\in \mathcal{P}(\mathcal{Y}\mid \mathcal{X})$ be a symmetric channel and $\gamma>0$.
    Write
    $Q(y) = \frac{1}{|\mathcal{X}|}\sum_{x\in \mathcal{X}}W(y\mid x)$, $y\in\mathcal{Y}$. 
    The quantities
    \begin{align}
        c_1(x) &= \sum_{y \in \mathcal{Y}}\min\{W(y\mid x), \gamma Q(y)\}, \label{eq:c1_constant}\\
        c_2(x) &= \sum_{y\in \mathcal{Y}}W(y\mid x) \mathbb{I}\left\{\frac{|\mathcal{X}|}{|\mathcal{X}(y)|}\leq \gamma\right\}, \label{eq:c2_constant}\\
        c_3(x) &= \sum_{y\in \mathcal{Y}} W(y\mid x) \frac{|\mathcal{X}|}{|\mathcal{X}(y)|}\mathbb{I}\left\{ \frac{|\mathcal{X}|}{|\mathcal{X}(y)|}\leq \gamma\right\}, \label{eq:c3_constant}\\
        r(x) &= \sum_{y\in \mathcal{Y}(x)}Q(y), \label{eq:r_constant}
    \end{align}
    are all independent of $x$.
\end{lemma}
\begin{proof}
    Since $W$ is symmetric, the output alphabet $\mathcal{Y}$ can be partitioned into disjoint subsets $\{\mathcal{Y}_1, \mathcal{Y}_2, \ldots, \mathcal{Y}_k\}$ such that, within each  $\mathcal{Y}_i$, the columns of the transition submatrix are permutations of one another. Therefore, 
    \begin{equation}
        Q(y_1) = \frac{1}{|\mathcal{X}|} \sum_{x\in \mathcal{X}}W(y_1\mid x) = \frac{1}{|\mathcal{X}|}\sum_{x\in \mathcal{X}}W(y_2\mid x) = Q(y_2), \quad \text{for any } y_1, y_2 \in \mathcal{Y}_i.
        \label{eq:constant_Q_proof}
    \end{equation}
    This implies that $Q$ is constant on each ${\cal Y}_i$; let $q_i$ denote this constant value. Then, for any $x_1, x_2 \in \mathcal{X}$,
    \begin{align*}
        c_1(x_1) &= \sum_{y\in \mathcal{Y}} \min\{W(y\mid x_1), \gamma Q(y)\} = \sum_{i=1}^k \sum_{y \in \mathcal{Y}_i} \min\{W(y\mid x_1), \gamma q_i\} \\
        &= \sum_{i=1}^k \sum_{y \in \mathcal{Y}_i} \min\{W(y\mid x_2), \gamma q_i\} = c_1(x_2),
    \end{align*}
    where the third equality follows because, within each subset $\mathcal{Y}_i$, the elements $\{W(y\mid x_1)\}_{y \in \mathcal{Y}_i}$ are a permutation of the elements $\{W(y\mid x_2)\}_{y \in \mathcal{Y}_i}$. 
    
    For $c_2(x)$ we have
    \begin{equation*}
        c_2(x) = \sum_{i=1}^k \sum_{y \in \mathcal{Y}_i} W(y\mid x) \mathbb{I}\left\{\frac{|\mathcal{X}|}{|\mathcal{X}(y)|}\leq \gamma\right\} = \sum_{i=1}^k I_i \sum_{y \in \mathcal{Y}_i} W(y\mid x),
    \end{equation*}
    where $I_i := \mathbb{I}\left\{ |\mathcal{X}|/|\mathcal{X}(y)| \leq \gamma \right\}$ is constant over $\mathcal{Y}_i$ because the columns of the submatrix are permutations (and therefore $|\mathcal{X}(y)|$ is constant for $y \in \mathcal{Y}_i$). Moreover, because the rows are permutation-invariant within each $\mathcal{Y}_i$, the sum $\sum_{y \in \mathcal{Y}_i} W(y\mid x)$ is also independent of $x$. This completes the proof for $c_2(x)$. The result for $c_3(x)$ follows along the same lines. 

    The independence of $r(x)$ can be argued as follows

    $$r(x) = \sum_{y \in \mathcal{Y}(x)} Q(y) = \sum_{i = 1}^k \sum_{y \in \mathcal{Y}(x)\cap \mathcal{Y}_i}Q(y) = \sum_{i = 1}^k |\mathcal{Y}(x)\cap \mathcal{Y}_i|q_i.$$

    For any $x_1\in \mathcal{X}$ and $x_2 \in \mathcal{X}$, the set of transition probabilities $\{W(y\mid x_1)\}_{y\in \mathcal{Y}_i}$ is a permutation of ${\{W(y\mid x_2)\}_{y\in \mathcal{Y}_i}}$ (row permutation). Therefore, the number of non-zero elements in $\mathcal{Y}(x)\cap \mathcal{Y}_i$ is independent of the input symbol, which combined with the expansion above of $r(x)$ gives the desired result.
\end{proof}

\begin{lemma}
\label{lemma:capacity_dispersion}
    Let $W \in \mathcal{P}(\mathcal{Y}\mid \mathcal{X})$ be a symmetric channel, $Q\in \mathcal{P}(\mathcal{Y})$ its capacity-achieving output p.m.f., and fix $\epsilon \in(0,1)$. The channel capacity, $C(W)$, and the channel dispersion, $V(W)$, can be expressed as:
    \begin{equation}
        C(W) = \mathbb{E}_{W_{x_0}}[i(x_0, Y)] \quad \text{and} \quad V_\epsilon(W) = \mathrm{Var}_{W_{x_0}}[i(x_0, Y)],
    \end{equation}
    for any arbitrarily chosen $x_0 \in \mathcal{X}$.
\end{lemma}
\begin{proof}
    Since the capacity-achieving output distribution is unique and the channel is symmetric, the distribution of $i(x, Y)$ under $W_x$ is independent of $x$. Consequently, 
     for every capacity-achieving input distribution $p\in\Pi$, the induced output distribution is $Q$. Moreover, symmetry implies that the conditional information-density distribution under $W_x$ is independent of $x$, with mean $C(W)$. Hence,
$$
U(p,W)
=\sum_xp(x)
  \mathbb E_{W_x}\!\left[(i(x,Y)-C(W))^2\right]
=\operatorname{Var}_{W_{x_0}}(i(x_0,Y)).
$$
Thus both the minimum and maximum of $U(p,W)$ over $p\in\Pi$ are equal. In particular, $V_\epsilon(W)$ is independent of $\epsilon\in(0,1)$.
\end{proof}

\begin{lemma}
    \label{lemma:capacity_ineq_for_general}
    Let $W \in \mathcal{P}(\mathcal{Y}\mid \mathcal{X})$ be a symmetric channel and $Q(y) = \frac{1}{|\mathcal{X}|}\sum_{x \in \mathcal{X}}W(y\mid x)$ is the capacity-achieving output distribution. Then, 

    $$D(W_{x}\|Q) \geq - \log Q(\mathcal{Y}(x)),$$
    for any arbitrary input symbol $x\in \mathcal{X}$.
\end{lemma}

\begin{proof}

The result follows easily from the log-sum inequality:

\begin{align*}
    D(W_x\|Q) &= \sum_{y \in \mathcal{Y}}W_x(y)\log \frac{W_x(y)}{Q(y)} = \sum_{y \in \mathcal{Y}(x)}W_x(y)\log \frac{W_x(y)}{Q(y)}\\
    &\geq \left(\sum_{y \in \mathcal{Y}(x)}W_x(y)\right) \log\frac{\left(\sum_{y \in \mathcal{Y}(x)}W_x(y)\right)}{\left(\sum_{y \in \mathcal{Y}(x)}Q(y)\right)} = \log \frac{1}{\sum_{y\in \mathcal{Y}(x)}Q(y)} = - \log Q(\mathcal{Y}(x)).
\end{align*}
\end{proof}


\begin{lemma}
\label{lem:new}
Let \(W\) be a symmetric channel with dispersion $V(W)>0$, and write
$$ b(W)= \begin{cases} \frac{1}{2},    &W\;\mbox{is nonsingular},\\ 
                                          0,    &W\;\mbox{is singular}. \end{cases}
$$
For every closed interval $E\subset(0,1)$, there are finite constants $K(E)$ and $n(E)$, depending only on $W$ and $E$, such that
$$
\log N_\alpha^*(W^{\otimes n}) \geq nC(W)+\sqrt{nV(W)}\Phi^{-1}(\alpha) +b(W)\log n-K(E),
$$
for every $\alpha\in E$ and all $n\geq n(E)$.
\end{lemma}

\begin{proof}
As we established in Lemma \ref{lemma:capacity_dispersion}, when the input distribution is uniform the information density has mean 
$C(W)$ and variance $V(W)$. Then the finite-blocklength coding-achievability proofs give bounds of the form
\begin{equation}
\log N_\alpha^*(W^{\otimes n}) \geq nC(W)+\sqrt{nV(W)}\,\Phi^{-1}(\alpha-\rho_n) +b(W)\log n-C,
\label{eq:oldbounds}
\end{equation}
where $C$ depends on $W$ but not on $\alpha$ and $n$, and where
$$0\leq \rho_n\leq \frac{D_W}{\sqrt n} ,$$
for all sufficiently large $n$, with the constant $D_W$ depending only on $W$.
Specifically, in the nonsingular case the lower bound on $N_\alpha^*$ follows from the proof of~\cite[Theorem~53]{polyanskiy:thesis}, 
which gives 
$$\rho_n=\frac{k_3+k_4}{\sqrt n}+e^{-k_1n},$$
and $C=\log k_2$, where the constants only depend on the input distribution and the channel. 
Under the uniform input distribution, nonsingularity implies that
$${\mathbb E}[\operatorname{Var}(i(X,Y)\mid Y)]>0,$$
so some output column contains two unequal positive transition probabilities, and both corresponding inputs have positive posterior probability.
In the singular case, the dependence-testing achievability proof in~\cite[Theorem 245 (248)]{PPV:10}
gives  the above the lower bound on $N_\alpha^*$ with 
$\rho_n=D_W/\sqrt n$, where $D_W$ only depends on $W$.

Now write $E=[\ell,u]$, for some $0<\ell\leq u<1$, and define
$$E'=\left[\frac{\ell}{2},\frac{1+u}{2}\right], \qquad D_E=\sup_{v\in E'}(\Phi^{-1})'(v) =\sup_{v\in E'}\frac{1}{\phi(\Phi^{-1}(v))}<\infty.$$
For large enough $n$, we have $\rho_n\leq\ell/2$, thus both $\alpha-\rho_n$ and $\alpha$ belong to $E'$, simultaneously for every $\alpha\in E$. 
Finally, the mean value theorem gives
\begin{align*} 
\sqrt{nV(W)} \bigl[\Phi^{-1}(\alpha)-\Phi^{-1}(\alpha-\rho_n)\bigr] 
&\leq \sqrt{nV(W)}D_E\rho_n\\ &\le \sqrt {V(W)}D_ED_W. \end{align*}
Substituting this into~\eqref{eq:oldbounds} completes the proof.
\end{proof}

\end{document}